\documentclass[lettersize,journal]{IEEEtran}
\usepackage[table]{xcolor}
\usepackage{cite}
\usepackage{graphicx}
\usepackage{lipsum}
\usepackage{stfloats}
\usepackage{bm}
\usepackage{amsmath,amssymb}
\usepackage{amsthm}
\usepackage{pifont}
\usepackage{bbding}
\usepackage{booktabs}
\usepackage{ctable}
\usepackage{threeparttable}
\usepackage{footmisc}
\usepackage{framed} 
\usepackage{array}
\usepackage{multirow}
\usepackage{longtable}
\usepackage{rotating}
\usepackage{fancyhdr}
\usepackage{lastpage}
\usepackage{layout}
\usepackage{wrapfig}
\usepackage{xcolor}
\usepackage{tabularx}
\usepackage{tabu}
\usepackage{enumitem}
\usepackage{url}
\usepackage{dsfont}
\usepackage{arydshln}

\usepackage{caption} 

\newtheorem{definition}{\textbf{Definition}}

\newtheorem{proposition}{\textbf{Proposition}}
\newtheorem{remark}{\textbf{Remark}}
\makeatletter
\renewenvironment{proof}[1][\proofname]{\par
  \pushQED{\qed}\normalfont\topsep6\p@\@plus6\p@\relax
  \trivlist\item[\hskip\labelsep
    \itshape
    #1\@addpunct{:}]\ignorespaces
}{%
  \popQED\endtrivlist\@endpefalse
}
\makeatother

\graphicspath{{Figures/}}
\DeclareGraphicsExtensions{.png, .eps}
\usepackage{soul}

\usepackage{acronym}
\acrodef{kb}[KB]{knowledge base} 
\acrodef{semcom}[SemCom]{semantic communication}
\acrodef{-d}[-D]{-dimensional}
\acrodef{ml}[ML]{machine learning}
\acrodef{snr}[SNR]{signal-to-noise ratio}
\acrodef{scs}[S-category space]{semantic-based category space}

\begin{document}
	\title{A Mathematical Framework of Semantic Communication based on Category Theory}
   \author{

	 Shuheng Hua,~\IEEEmembership{Graduate Student~Member,~IEEE},
     Yao Sun,~\IEEEmembership{Senior~Member,~IEEE},
      Kairong Ma,~\IEEEmembership{Graduate Student~Member,~IEEE},
       Dusit Niyato,~\IEEEmembership{Fellow,~IEEE},
     and Muhammad Ali Imran,~\IEEEmembership{Fellow,~IEEE}

	\thanks{
	
	Shuheng Hua, Kairong Ma, Muhammad Ali Imran, Yao Sun (\textit{Corresponding author: Yao Sun}.) are with the James Watt School of Engineering, University of Glasgow, Glasgow G12 8QQ, UK (e-mail: s.hua.1@research.gla.ac.uk; \{Yao.Sun, Muhammad.Imran\}@glasgow.ac.uk; k.ma.2@research.gla.ac.uk.)

    Dusit Niyato is with the College of Computing and Data Science, Nanyang Technological University, Singapore 639798 (e-mail: dniyato@ntu.edu.sg).
  }	}
  
	\maketitle 
	\begin{abstract}
While semantic communication (SemCom) has recently demonstrated great potential to enhance transmission efficiency and reliability by leveraging machine learning (ML) and knowledge base (KB), there is a lack of mathematical modeling to rigorously characterize SemCom system and quantify the performance gain obtained from ML and KB. 
In this paper, we develop a mathematical framework for SemCom based on category theory, rigorously modeling the concepts of semantic entities and semantic probability space. 
Within this framework, we introduce the semantic entropy to quantify the uncertainty of semantic entities. 
We theoretically prove that semantic entropy can be effectively reduced by exploiting KBs, which capture semantic dependencies. 
Within the formulated semantic space, semantic entities can be combined according to the required semantic ambiguity, and the combined entites can be encoded based on semantic dependencies obtained from KB.
Then, we derive semantic channel capacity modeling, which incorporates the mutual information obtained in KB to accurately measure the transmission efficiency of SemCom.
Numerical simulations validate the effectiveness of the proposed framework, showing that SemCom with KB integration outperforms traditional communication in both entropy reduction and coding efficiency.

	\end{abstract}
	

	
	\section{Introduction}

\Ac{semcom} has emerged as a transformative paradigm to enhance data transmission efficiency and robustness based on advanced \ac{ml} models \cite{15}. Unlike the traditional communication framework built on Shannon's information theory \cite{1}, which focuses mainly on the accuracy of bits delivery, \ac{semcom} turns the focus from bits to semantics, emphasizing the interpretation and understanding of the content being communicated \cite{3}. This paradigm considers not only the raw data, but also the context, relevance and meaning of the data, aiming to reduce the number of bits required for transmission \cite{12}. 

\subsection{Background and Related Works}
Generally, based on the task requirements, the sender employs a deep learning (DL)-based encoder to extract key features from the original data and map them into a high-dimensional vector representation \cite{4}. These vectors capture the underlying relationships between the data elements, while integrating the \ac{kb} to enhance context and filter out irrelevant information \cite{16}. The decoder at the receiver side performs semantic inference and multi-modal reconstruction based on the vector distribution, leveraging the local \ac{kb} to further reconstruction \cite{17}.

Recently, most of works \cite{22,23,24,25,26,29,30,31,32,33} focus on the use of \ac{ml} models to enhance semantic extraction and reconstruction capabilities. Compared with traditional communication paradigms, these data-driven approaches show greater adaptability in dealing with complex semantic relationships. Meanwhile,more and more works are trying to integrate \ac{kb} into \ac{semcom} to enhance semantic reasoning. The work of \cite{10} propose a knowledge graph-based \ac{semcom} system that uses a deep learning (DL) model to convert transmitted sentences into triples in order to improve the accuracy of semantic representation.
The authors of \cite{8} propose a generative \ac{semcom} architecture based on semantic \ac{kb}, which achieves low-dimensional semantic characterization by constructing three sub-\ac{kb}s, namely, source, task and channel, to significantly improve the efficiency of 6G communication.
Besides, the authors of \cite{11} proposed a \ac{semcom} system empowered by DL-based on a shared \ac{kb}, aiming to reduce the amount of transmitted symbols by combining text messages and corresponding knowledge while maintaining semantic performance. In \cite{18}, the authors present an contrastive representations learning based \ac{semcom} framework (CRLSC), which constructs a shared \ac{kb} through a large pre-trained model on the server side, which is utilized by the end devices for training, facilitating knowledge transfer in large-scale heterogeneous networks. 
Apart from these, the authors of \cite{19} proposed a knowledge-enhanced \ac{semcom} reception framework that enables the receiver to use the \ac{kb} for semantic reasoning and decoding, improving the reliability and performance of the system. 

While the above works able to effectively exploit \ac{ml} and \ac{kb} to enhance \ac{semcom} system performance, there is still a lack of fundamental mathematical framework to characterize the \ac{semcom} system and provide theoretically guidance to optimally design the system. 
Such a framework should be essential for understanding and underlying the design principles of \ac{semcom}, enabling rigorous performance analysis and evaluation. There are only a few works on establishing mathematical framework for \ac{semcom}. 
A definition of semantic entropy based on logical probability in \cite{20} has been proposed in a first attempt to quantify the semantic information in sentences. 
The authors of \cite{27} extended the semantic entropy model by combining propositional logic and statistical probability to introduce a semantic entropy formula for \ac{kb} association. 
In \cite{14}, the authors proposed a framework of semantic information theory (SIT), centered on constructing equivalence classes between semantics and syntax by mapping the same semantic information to multiple syntactic representations. 
The work in \cite{28} constructs a theoretical framework for \ac{semcom}, focuses on optimizing the existing semantic language (language utilization) and designing efficient communication protocols (language design), and reveals the optimization paths for meaning transfer through probabilistic models and distortion-cost analysis.

\subsection{Motivation and Challenges}
As discussed above, these theories have not explored how to characterize SemCom system and quantify the performance gain obtained from and \ac{kb}. 
Specifically, existing frameworks lack mathematical definitions of semantics and effectively link the \ac{kb} to the reduction of semantic entropy, which is an essential feature for improving communication efficiency. 
This gap presents several interconnected challenges: First, the quantification of relationships between semantic entities is a complex problem. Traditional mathematical frameworks (e.g., graph theory and probability theory) are often difficult to capture the diversity and complexity of semantic relationships \cite{34}. 
It is hard for these traditional approaches to model the multidimensional properties of semantic connections, including synonyms, antonyms and contextual relationship. 
Second, the problem of constructing and describing semantic space is also a difficult problem to be solved. As a multidimensional abstract space for the existence of semantic entities, semantic space needs to be precisely defined in terms of its dimensionality, metrics, and topological properties. Existing vector space models represent relationships in semantic space by mapping semantic entities to vectors. 
However, these models still have major limitations when facing complex semantics. Third, the deep integration of \ac{kb} and semantic entropy has not yet been systematically modeled. Specifically, how to deeply integrate the multidimensional features of \ac{kb} with the process of semantic entropy reduction, so as to support adaptive semantic coding strategies (e.g., dynamic adjustment of synonym set priority) and knowledge-driven resource allocation (e.g., bandwidth optimization based on semantic entropy) is not only of great theoretical significance, but also will provide a solid foundation for the practical application of \ac{semcom}.
 
In order to address these challenges, we need to delve into the following key problems:
\begin{itemize}
    \item \textit{Problem 1: How should semantic entities be represented?} 
    This addresses the challenge of semantic relationship quantification by seeking appropriate mathematical representations for semantic entities that capture their multidimensional relationships.
\end{itemize}
\begin{itemize}
    \item \textit{Problem 2: How to measure the performance gain from KB?} 
    This builds on both the first and third challenges, seeking metrics that can quantify how knowledge bases enhance semantic communication efficiency.
\end{itemize}
 
\begin{itemize}
    \item \textit{Problem 3: How semantic entropy quantifies the relationship between compression rate and semantic ambiguity?} 
    This directly corresponds to the KB-Semantic entropy integration challenge, exploring the mathematical framework needed to relate information compression with semantic clarity.
\end{itemize}

\subsection{Contributions and Organization}
In response to the problems outlined above, in this paper, we make use of category theory for defining and organizing relationships between semantic entities. Specifically, it allows us to model the composition and transformation of semantic structures in a way that can be easily combined with information theoretic principles. In addition, we address the compositional nature of semantics - \textit{i.e.}, how simple semantic units can be combined to form more complex meanings - by utilizing concepts such as functions and morphisms in category theory. The main contributions of this paper are summarized as follows:
\begin{itemize}
    \item We propose a mathematical framework for the semantics of communication systems using category theory. It formally defines semantic entities as tuples of attributes from multiple categories and organizes them into semantic probability space, which is multi-dimensional space representing entity relationships and attributes.
\end{itemize}

\begin{itemize}
  \item  We distinguish between conventional entropy and semantic entropy and show how semantic entropy can be reduced by combining related or synonymous attributes in conjunction with the \ac{kb}.
\end{itemize}

\begin{itemize}
    \item We explore how information gain from the \ac{kb} can be obtained using Kullback-Leibler divergence, which measures the reduction in uncertainty achieved by using the \ac{kb}. This demonstration shows that the use of \ac{kb} reduces information entropy, allowing for more efficient communication with fewer resources.
\end{itemize}

\begin{itemize}
    \item We formalize information-level entropy by defining how to compute information entropy consisting of multiple entities, each belonging to a semantic probability space. By considering the relationships between entities, we propose a KB-based wireless semantic channel capacity model for \ac{semcom}.
\end{itemize}

\begin{itemize}
    \item We verify the effectiveness of the established mathematical framework through numerical simulations. We compare our semantic coding approach (utilizing Fano coding based on spatial location and probabilistic information) with two benchmark approaches (traditional Fano coding and Fano coding with parity check). The results also verify the superiority of the performance of semantic coding with KB implementation in different situations.
\end{itemize}

The rest of this paper is organized as follows. The semantic entities and semantic probability space are defined Section \uppercase\expandafter{\romannumeral2}. Section \uppercase\expandafter{\romannumeral3} illustrates entropy reduction at entity level. In Section \uppercase\expandafter{\romannumeral4}, the entropy reduction extend to message level is discussed. Numerical results are demonstrated and discussed in Section \uppercase\expandafter{\romannumeral5}, followed by the conclusions in Section \uppercase\expandafter{\romannumeral6}. 






\section{Semantic Entity \& Semantic Probability Space}
Borrowing from linguistic science, semantics is defined as how language conveys meaning through words, phrases, and symbols and their relationship to what they represent \cite{21}. To mathematically characterize semantics, we introduce the concept of a semantic entity, which serves as the fundamental unit to constrcut semantics. A semantic entity refers to an object, concept, or thing that has meaning within a particular context. For example, the character `Tom', from the cartoon series \textit{Tom \& Jerry}, can be a semantic entity, which with precise representations: his blue color, his feline kind, and his role as a cartoon character. These representations inspire us to exploit category theory to formally define a semantic entity and semantic probability space.

\subsection{Semantic Probability Space}
We start by defining a general category space and then narrow down to specific formulations of semantic entities and semantic probability space.


\subsubsection{{Category Space}} 
Consider a discrete information source, 
and let \(\mathcal{B}\) represents the set of elements:
\begin{equation}
    \mathcal{B}=\{b_1,\ldots,b_N\}, 
\end{equation}
with probabilities \(\{ p(b_1),\ldots,p(b_N) \}\).

\begin{definition}
    A \textbf{category} \(\mathcal{C}_j\) is a partition of a set: 
    \begin{equation}
    \mathcal{C}_j: \mathcal{B}\rightarrow\mathcal{B}_j= \{B_j^0,B_j^1,B_j^2,\ldots,B_j^\mathit{M_j}\}, 
\end{equation}
where the divided subsets should fulfill 
\begin{equation} \label{excclusive}
    B_j^i\cap B_j^k = \emptyset, \forall i \ne k,
\end{equation}
and 
\begin{equation}\label{cup_all}
     B_j^0 \cup B_j^1\cup\cdots\cup B_j^{M_j}= \mathcal{B}.
\end{equation}
\end{definition}

Denote \(\mathcal{C}_j^i\) as the \(j_\mathit{th}\) sub-mapping of the partition \(\mathcal{C}_j\), that \(\mathcal{C}_j^i(\mathcal{B})=B_j^i\). 
Denote \(\circ\) as composition operator, which combines two mappings such that the output of one mapping becomes the input of the other. Therefore, we have
\begin{equation}
    (\mathcal{C}_j^{j'}\circ\mathcal{C}_k^{k'})(\mathcal{B})=\mathcal{C}_j^{j'}(\mathcal{C}_k^{k'}(\mathcal{B}))=\mathcal{C}_j^{j'}(B_k^{k'})=B_j^{j'}\cap B_k^{k'}.
\end{equation}
Then we have commutativity: 
\begin{equation} \label{commutative}
    \mathcal{C}_j^{j'}\circ\mathcal{C}_k^{k'}=\mathcal{C}_k^{k'}\circ\mathcal{C}_j^{j'}; 
\end{equation}
associativity:
\begin{equation} \label{associativity}
    \mathcal{C}_j^{j'}\circ(\mathcal{C}_k^{k'}\circ\mathcal{C}_l^{l'})= (\mathcal{C}_j^{j'}\circ\mathcal{C}_k^{k'})\circ\mathcal{C}_l^{l'}; 
\end{equation}
and distributivity:
\begin{equation} \label{distributivity}
    \mathcal{C}_j^{j'}\circ(\mathcal{C}_k^{k'}\cup\mathcal{C}_l^{l'})= (\mathcal{C}_j^{j'}\circ\mathcal{C}_k^{k'})\cup(\mathcal{C}_j^{j'}\circ\mathcal{C}_l^{l'}).
\end{equation}

Given \(M\) categories \(\mathcal{C}_j\) for a set \(\mathcal{B}\), we can obtain the discrete category space: 
\begin{equation}
    \mathcal{C}=(\mathrm{Obj}(\mathcal{C}),\mathrm{Hom}(\mathcal{C}),\circ), 
\end{equation}
where \(\mathrm{Obj}(\mathcal{C})\) represents the set of \(\mathcal{B}\) and all its subsets, \(\mathrm{Hom}(\mathcal{C})\) represents the set of all mappings.

Note that any partition of a set \(\mathcal{B}\) that satisfies (\ref{excclusive}) and (\ref{cup_all}) can be defined as a category. 
Then, partition the set \(\mathcal{B}\) based on the semantic meaning of elements, a \ac{scs} can be obtained. 

\subsubsection{S-category Space} 
Categorizing \(\mathcal{B}\) based on the semantic meaning of its elements, we can obtain a semantic category\footnote{In previous subsection, \(\mathcal{C}_j\) denoted an arbitrary partition. To reduce symbol redundancy, \(\mathcal{C}_j\) represents a semantic partition from here.}:
\begin{equation}
     \mathcal{C}_j: \mathcal{B}\rightarrow \mathcal{B}_j= \{B_j^0,B_j^{a_j^1},B_j^{a_j^2},\ldots,B_j^{a_j^{M_j}}\},
\end{equation}
where \(B_j^0\) is the set of elements which are not relevant to semantic category \(j\). 
For example, if a semantic category \(\mathcal{C}_j\) is color, then we can divide all elements based on color. 
In this case, \(B_j^0\) is the set of elements which are not relevant to color, 
\(B_j^{a_j^1}\) could be the set of elements which are relevant to red, 
\(B_j^{a_j^2}\) is the set of elements which are relevant to blue. 
Therefore, we obtain a semantic partition based on a color category: 
\begin{equation}
    \mathcal{C}_{\mathrm{color}}: \mathcal{B}\rightarrow \mathcal{B}_{\mathrm{color}}= \{B_{\mathrm{color}}^0,B_{\mathrm{color}}^{\mathrm{red}},B_{\mathrm{color}}^{\mathrm{blue}},\ldots,B_{\mathrm{color}}^{\mathrm{green}}\}. 
\end{equation}
Furthermore, depending on user needs, the category \(\mathcal{C}_j\) can be extended to more detailed categories, i.e., \(\mathcal{C}_j\rightarrow \mathcal{C}_j^{(k)}\). 
For example, `color' can be extended to `RGB': 
\begin{equation}
    \mathcal{C}_\mathrm{color}\rightarrow \mathcal{C}_R\times \mathcal{C}_G\times \mathcal{C}_B.
\end{equation}
Then, the following step is to specify how many elements are required for a given semantic category. 
For example, how many colors can we define? 
We could consider green, red, blue as subsets within a color partition, and we can also include white, black, etc. 
Therefore, let us introduce the concept of ``attribute" now. 

\begin{figure*}[htbp]
\centering
    \includegraphics[scale=0.55]{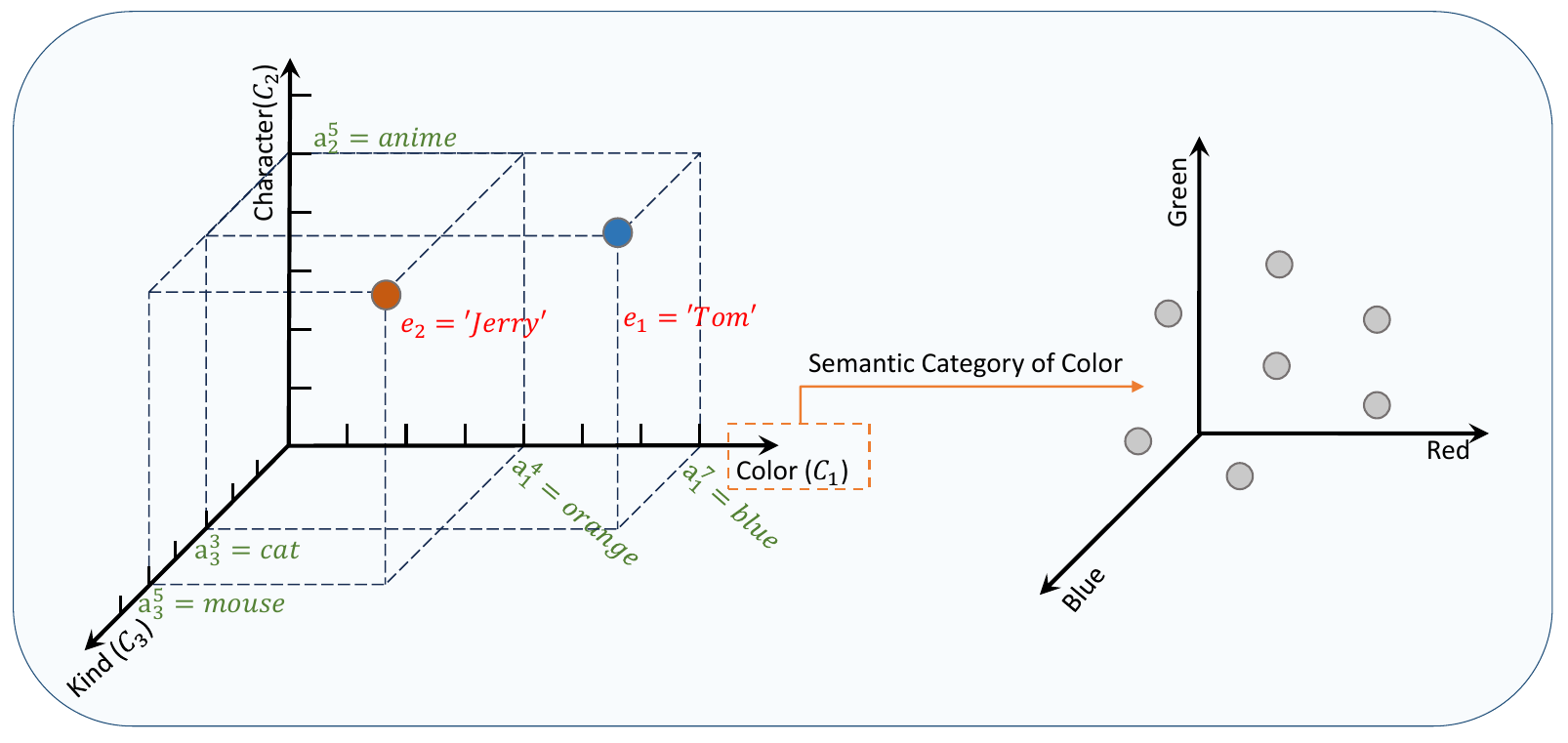}
    \caption{The location of entity in \ac{scs}}
    \label{semantic_space}
\end{figure*}

\begin{definition} \label{def_attribute}
    For a semantic category \(\mathcal{C}_j: \mathcal{B}\rightarrow \mathcal{B}_j= \{B_j^0,B_j^{a_j^1},B_j^{a_j^2},\ldots,B_j^{a_j^{M_j}}\}\), an \textbf{attribute} of \(\mathcal{C}_j\) is defined as 
    \begin{equation}
        a_j^i: \{a_j^i|B_j^{a_j^i}\in\mathcal{B}_j\}.
    \end{equation}
\end{definition}
Attributes are properties or characteristics that provide specific details, for a partitioned subset. 
For example, `blue' could be an attribute of `color' category, and the subset \(B_{\mathrm{color}}^{\mathrm{blue}}\) with the superscript `blue' contains all elements relevant to `blue' in \(\mathcal{B}\).
Condition (\ref{excclusive}) ensures that the attributes of a category are exclusive, which avoid misunderstanding of semantics. 
{Let} \(A_j\) represent the set of attributes of category \(\mathcal{C}_j\): 
    \begin{equation}
        A_j=\{a_j^1,a_j^2,\ldots,a_j^{M_j}\}, 
    \end{equation}
where \(M_j\) denotes the number of attributes in \(A_j\).

We can combine multiple semantic categories for a given discrete information source. 
For example, for given names of cartoon characters, they can be categorized in color category and kind category, as shown in Fig. \ref{semantic_space}. We can see that there are three semantic categories, color, character, and kind. Each of these categories represents an independent dimension in the \ac{scs}, allowing entities (such as `Tom' or `Jerry') to be precisely located based on their attributes.
Therefore, given \(M\) semantic categories, we can obtain an \(M\)-dimensional (\(M\)-D) \ac{scs}: 
\begin{equation}
    \mathcal{C}_s^{(M)}=(\mathrm{Obj}(\mathcal{C}_s^{(M)}),\mathrm{Hom}(\mathcal{C}_s^{(M)}),\circ).
\end{equation}

Now, based on the concept of \ac{scs}, let us give the formal definition of a semantic entity.
\begin{definition}\label{def_entity}
    In an \(M\)-\(\mathrm{D}\) \ac{scs}, an \(m\)-\(\mathrm{D}\) \textbf{semantic entity} is defined as 
    \begin{equation}
        e_i^{(m)}=(a_1^{i_1},a_2^{i_2},\ldots,a_m^{i_m}), m\leq M.
    \end{equation}
\end{definition}
Note that if \(m=M\), the entity is a point in the \ac{scs}; 
if \(m<M\), the entity is an \((M-m)\)-D geometry orthogonal to the axes. 
For example, in the 3-D \ac{scs} shown as Fig. \ref{semantic_space}, the 1-D entity `cat' is a 2-D plane determined by \((C_3=\mathrm{cat})\); 
the 2-D entity `blue cat' is a line determined by \((C_1=\mathrm{blue},C_3=\mathrm{cat})\); the 3-D entity `Tom' is a point. 
These low-dimensional entities \(e_i^{(m)},\forall m<M\) will be projected to a point in one of the \(m\)-D {projection subspaces}. 
{In this work, we use orthogonal projection point to represent the entities} \(e_i^{(m)},\forall m<M\), \textit{i.e.}, setting the attribute values to 0 for the irrelevant categories. 
For example, the entity \(e_i^{(2)}=(\mathrm{blue},0,\mathrm{cat})\) represents `blue cat' with an unspecified character, in the aforementioned  \(3\)-D \ac{scs}. 
From the Definition \ref{def_attribute} and \ref{def_entity}, {attribute is the minimum semantic unit (sut) of semantics} in a \ac{scs}. 
{Then, the amount of semantics (in the unit of suts) of an entity} \(e_i^{(m)}\) can be calculated as:
\begin{equation}
    q(e_i^{(m)})=||e_i^{(m)}||_0
\end{equation}



\subsubsection{Semantic Probability Space} 
After establishing a \ac{scs}, we now should answer other questions, i.e,. how to measure the difference for two semantic entities, {and how to calculate the semantic entropy.} 
Therefore, for convenience in calculation, we introduce an Euclidean semantic space \(\Omega^{(M)}= \mathbb{R}^M\), {with the probability information of source}. 

We decompose the measurement of semantic entities distance into the measurement of attributes distance. 
Let \(d_{C_j}(\cdot)\) denotes the semantic distance function of attributes of \(\mathcal{C}_j\): 
\begin{equation}
    \begin{aligned}
    d_{\mathcal{C}_j}(a_j^i,a_j^k)
    \left\{
    \begin{array}{ll}
      > 0, &\mathrm{if}\quad i\ne k, \\
      = 0, &\mathrm{if}\quad i=k.
    \end{array}
    \right.
    \end{aligned}
\end{equation}
The attributes in \(A_j\) can be mapped to \(j_{\mathrm{th}}\) axis in \(\Omega^{(M)}\) with the mapping function \(f_{C_j}\): 
\begin{equation}
    f_{C_j}: a_j^i \rightarrow x_j^i, \forall a_j^i\in A_j, x_j^i \in \mathbb{R}.
\end{equation}
Let \(a_j^i\) and \(a_j^{k}\) be two attributes in \(A_j\), we define they are \(\epsilon\)-similar, if 
\begin{equation}\label{similar}
    0<|f_{C_j}(a_j^i)-f_{C_j}(a_j^{k})|_{}= d_{C_j}(a_j^i,a_j^{k})\leq \epsilon,
\end{equation}
where \(\epsilon>0\) is a similarity threshold. 
\\If \(\exists a_j^{k'} \in A_j\) is the antonym of \(a_j^k\in A_j\), 
then we have: 
\begin{equation}\label{antonym}
    f_{C_j}(a_j^k)=-f_{C_j}(a_j^{k'}). 
\end{equation}

In this content, the semantic entity can be represented with real coordinate in \(\Omega^{(M)}\):
\begin{equation}
\begin{aligned}
         f_{\Omega}(\cdot): e_i^{(M)} \rightarrow (x_{1}^{i_1},x_{2}^{i_2},\ldots,x_{M}^{i_M}).
\end{aligned}
\end{equation}

Thus, the difference for two semantic entities could be the distance in the Euclidean space\footnote{In this work, we assume the different categories have the same importance. If the importance of categories is considered vary, importance weights need to be applied accordingly.}: 
\begin{equation}
    d_{\Omega}(e_i^{(M)},e_j^{(M)})=|| f_{\Omega}(e_i^{(M)})- f_{\Omega}(e_j^{(M)})||_p, 
\end{equation}
the \(p\)-norm distance can be considered as \(1\)-norm Manhattan distance, \(2\)-norm Euclidean distance, or other norms, as per the practical requirement. 



\begin{figure*}
\centering
    \includegraphics[scale=0.59]{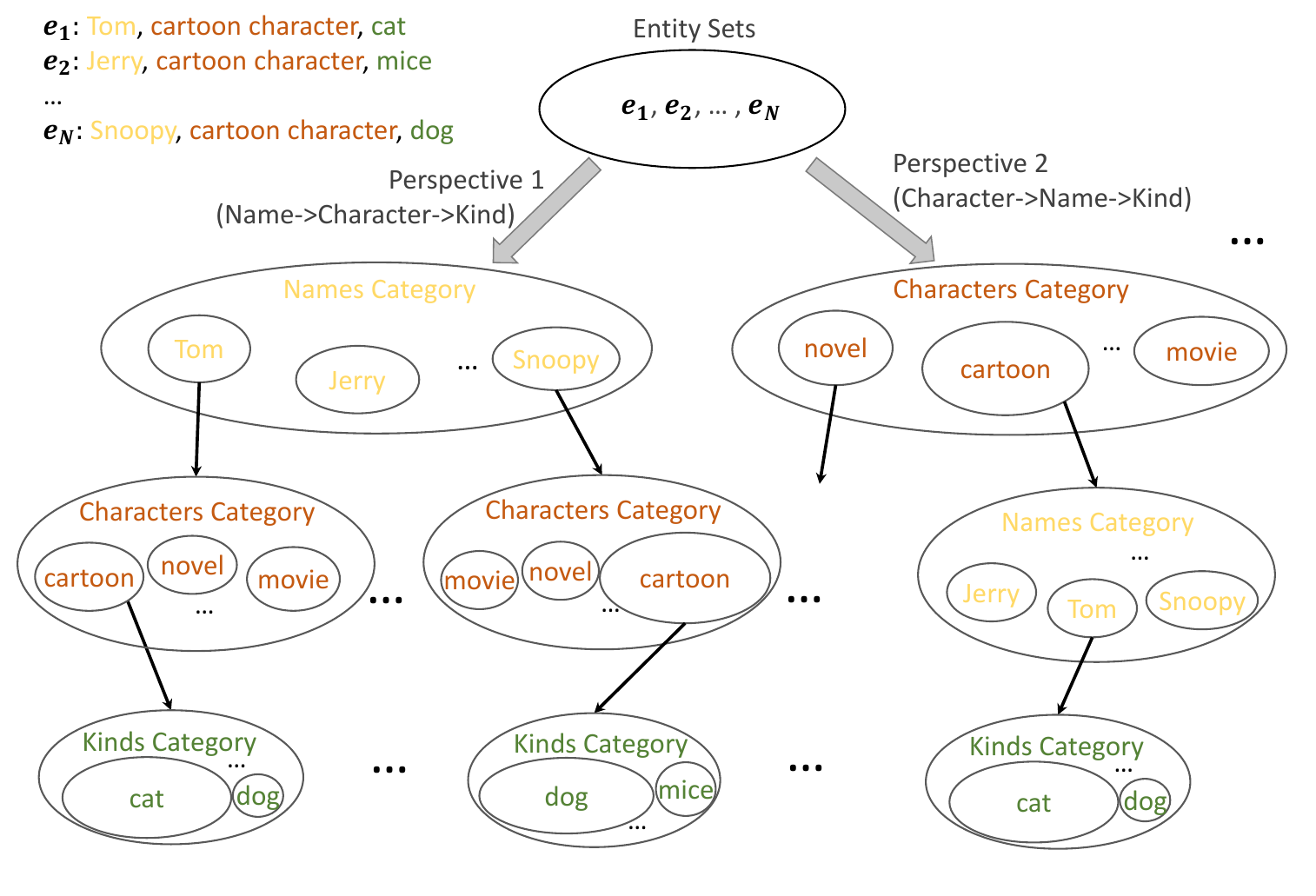}
    \caption{Different perspectives for locating an entity}
    \label{perspectives}
\end{figure*}

Based on these mapping, we can exactly know the position of an entity in the Euclidean semantic space. 
By considering probability information of entities in the space, the semantic probability space \((\Omega^{(M)},E,P)\) can be established, 
where \(E\) is the set of entities, and \(P\) is a probability function, which reflects the probability of each entity in the space. 
The measure of the entire sample space is equal to one, \textit{i.e.}, \(P(\Omega)=1\).
Then, given the attributes of \(m<M\) categories, we can determine an \(\mathit{(M-m)}\)-D subspace \(\Omega^{(M-m)}\). 
The probability of entities in this \(\Omega^{(M-m)}\) should be calculated based on conditional probability, i.e., 
\begin{equation}
    P(e^{(M-m)}|\Omega^{(M-m)}) =\frac{P(e^{(M)}|\Omega^{(M)})}{P(\Omega^{(M-m)}|\Omega^{(M)})}.
\end{equation}
Thus, the measure of the new sample space keeps equal to one:  
\begin{equation} \label{sub-space}
    P(\Omega^{(M-m)})=1.
\end{equation}

In the semantic probability space, there are many different perspectives (i.e., the orders of categories) to locate an entity, as shown in Fig. \ref{perspectives}. 
Taking the example of `Tom', there are different perspectives shown in Fig. \ref{perspectives}: one is \((\text{Name} \rightarrow \text{Character} \rightarrow \text{Kind})\), another one is \((\text{Character} \rightarrow \text{Name} \rightarrow \text{Kind})\), and other perspectives. When we mention `Tom' (name) in `cartoon' (character), 
we can determine the entity with {high confidence}. 
This means, the entity \(e_1\) has the dominant probability in the subspace determined by (name = Tom, character = cartoon character). 
While if we locate the entity with `cat' and `cartoon', these two attributes may not enough. 
Commutativity (\ref{commutative}) and associativity (\ref{associativity}) ensure that from different perspective to locate an entity is the same.

\subsection{Categorizing Entropy}

In traditional communication, we calculate entropy of information source \(\mathcal{B}\) with the probabilities of elements \(p(b_i)\) as: 
\begin{equation}
    H_c(\mathcal{B})=\sum_{b_i\in \mathcal{B}} -p(b_i)\log_2 p(b_i).
\end{equation}

For a given semantic probability space \((\Omega^{(M)},E,P)\) and a perspective (\(\mathcal{C}_1\rightarrow\mathcal{C}_2\rightarrow\cdots\rightarrow\mathcal{C}_M\)), 
the self-information of an entity \(e_i^{(M)}=(a_1,\ldots,a_M)\) can be expressed as: \(I(e_i^{(M)})=\sum_{j=1}^{M} I(a_j|a_1,\ldots,a_{j-1})\),
where each term \(I(a_j|a_1,\ldots,a_{j-1})\) follows Shannon’s definition of self-information, reflecting the degree of surprise (information content) contributed by \(a_j\) given the preceding attributes. 
Then, the categorizing entropy of entities \(E\) can be calculated as the sum of the entropy of semantic partitions:
\begin{equation}
\begin{aligned}
     H_s({E})=& \sum_{a_1^{j}\in A_1}-p(a_1^j)\log_2p(a_1^j) \\
    +&\sum_{a_1^j\in A_1}\left[ p(a_1^j) \sum_{a_2^j\in A_2}-p(a_2^j|a_1^j)\log_2p(a_2^j|a_1^j) \right] \\
    +&\cdots\\
    +&\sum_{a_1^j\in A_1} \cdot\cdot\cdot \left. \sum_{a_{N-1}^j\in A_{M-1}} \right[ p(a_{M-1}^j|a_1^j...a_{M-2}^j) \\
    & \left. \sum_{a_M^j\in A_M}-  p(a_M^j|a_1^j...a_{M-1}^j)\log_2p(a_M^j|a_1^j...a_{M-1}^j) \right],
\end{aligned}
\end{equation}
where \(
p(a_{2}^{j}|a_{1}^{j}) =
\frac{
    \textstyle \sum_{b_i \!\in\! B_{1}^{a_{1}^{j}} \cap B_{2}^{a_{2}^{j}}}  p(b_i)
}{
    \textstyle \sum_{b_i \!\in\! B_{1}^{a_{1}^{j}}}  p(b_i)
}
\)
denotes the probability of subset relevant to \(a_2^j\) in subspace determined by \(a_1^j\). 

Consider a semantic probability space \(\Omega^{(M)}\) of information source \(\mathcal{B}\), 
each element \(b_i \in \mathcal{B}\) uniquely corresponds to an entity in the semantic probability space, and each entity corresponds to at most one element: 
\(\forall b_i\in \mathcal{B}\), \(\exists e_{i'}\in E\), let \(b_i\in e_{i'}\), and \(\forall j\neq i\), we have \(b_j\notin e_{i'}\). 
Thus, the probabilities of entities can be: 
\begin{equation}\label{detailed_p}
    \begin{aligned}
    p(e_{i'})=
    \left\{
    \begin{array}{ll}
      p(b_i), &\mathrm{if}\quad \exists  b_i\in\mathcal{B}, b_i\in e_{i'};\\
      = 0, &\mathrm{if}\quad \forall b_i\in\mathcal{B}, b_i\notin e_{i'}.
    \end{array}
    \right.
    \end{aligned}
\end{equation}
In such a semantic probability space, we have Proposition 1, i.e., the categorizing entropy of the entities is the same as the entropy of the source. 

\begin{proposition}
    For a given discrete information source \(\mathcal{B}\) and a semantic probability space with the set of entities \(E\), if each element \(b_i\in\mathcal{B}\) uniquely corresponds to an entity \(e_i\in E\), 
    the categorizing entropy of entities is equal to classical entropy of the source, \textit{i.e.}, 
    \begin{equation}
        H_s(E)=H_c(\mathcal{B}).
    \end{equation}
\end{proposition}
\begin{proof}

Consider the information source with probabilities \(P_b=\{x_1,x_2,\ldots,x_i,\ldots,x_N\}\), 
\begin{equation}
    \sum_{i=1}^{N} x_i=1. 
\end{equation}
Step 1. Consider a 2-D semantic probability space, with the number of attributes \(\mathrm{size}(A_1)=M_1\), \(\mathrm{size}(A_2)=M_2\), \(M_1\times M_2 \geq N\).
Let \(y_j\) denote the probability of attribute \(a_1^j\in A_1\), 
\begin{equation}
   P(a_1^j)=y_j;\quad \sum_{j=1}^{M_1} y_j=1. 
\end{equation}
Let \(z_j^{i}\) denote the probability of attribute \(a_2^i\in A_2\) in subspace determined by \(a_1^j\): \(P(a_2^i|a_1^j)=z_j^i;\) and from (\ref{sub-space}) we have: 
\begin{equation}\label{sum_xj=1}
    \sum_{i=1}^{M_2}z_j^{i}=1,\forall j \in \mathcal{M}_1.
\end{equation}
The categorizing entropy is: 
\begin{equation} \label{s_entropy_level1}
    H_s(E)=-(\sum_{j=1}^{M_1} y_j\log_2 y_j + \sum_{j=1}^{M_1} y_j \sum_{i=1}^{M_2} z_j^{i} \log_2 z_j^{i}).
\end{equation}

Substitute (\ref{sum_xj=1}) into equation (\ref{s_entropy_level1}), 
    \begin{equation}
    \begin{split}
        H_s(E)&=-(\sum_{j=1}^{M_1} y_j\sum_{i=1}^{{M_2}}z_j^{i} \log_2 y_j + \sum_{j=1}^{M_1} y_j \sum_{i=1}^{M_2} z_j^{i} \log_2 z_j^{i})\\
        &=-(\sum_{j=1}^{M_1} \sum_{i=1}^{M_2} y_j z_j^{i} \log_2 y_j + \sum_{j=1}^{M_1}  \sum_{i=1}^{M_2} y_j z_j^{i} \log_2 z_j^{i})\\
        &=-\sum_{j=1}^{M_1} \sum_{i=1}^{M_2}y_j z_j^{i} \log_2 y_j z_j^{i}.
    \end{split}
    \end{equation}
   The set of probabilities of entities is: 
   \begin{equation}
       {P}_{e}^{(2)}=\{\ldots,y_j z_j^{i},\ldots\},\forall j\in \mathcal{M}_1,i\in\mathcal{M}_2. 
   \end{equation}
   From (\ref{detailed_p}) we have 
   \begin{equation}
       P_e^{(2)}\setminus \{0\}={P}_{b}.
   \end{equation}
   Therefore, for this 2-D semantic probability space, we have
    \begin{equation}
         H_s(E)=H_c(\mathcal{B}).
    \end{equation} 
    
Step 2. Then, extending the same derivations to $M$-D semantic probability space: 
    \begin{equation}
       {P}_{e}^{(M)}=\{\ldots,y_j z_j^{i} w_{ji}^k\cdots,\ldots\},\forall j\in \mathcal{M}_1,\ldots. 
   \end{equation}
      \begin{equation}
       P_e^{(M)}\setminus \{0\}={P}_{b}.
   \end{equation}
   Thus, 
    \begin{equation}
         H_s(E)=H_c(\mathcal{B}).
    \end{equation}
\end{proof}

Proposition 1 theoretically prove the intuitive expectation, i.e., partitioning the elements of a source into entities in a semantic probability space of any dimension, with a one-to-one correspondence, does not change the entropy. 
This proposition verifies the validity of the formulated semantic probability space. In the following, we will demonstrate how this formulated space enables the mathematical modeling of \ac{kb} integration in \ac{semcom} and facilitates quantifying the performance gain obtained from \ac{kb}.

\section{Entropy Reduction at Entity Level} 
In this section, we discuss how KB can help SemCom at the individual entity level based on the mathematical framework established in Section \uppercase\expandafter{\romannumeral2}. 
Basically, there would be two ways to reduce the entropy by utilizing \ac{kb}. 
The first way is to combine synonyms based on \ac{kb} to reduce the entropy of source information with no semantic ambiguity. 
The second way is to exploit semantic dependencies between different entities and combine entities with similar semantic meanings under a tolerable ambiguity requirement. 

For better clarity and understanding, we first provide a formal definition of \ac{kb}. 
\begin{definition}
    For a set of entities \(E\), the \ac{kb} is the conditional probability distribution of \(k_\mathrm{th}\) entity, 
    given \(l\)-length entity sequence \((e^1,\ldots,e^{l})\),\footnote{In this work, \(e^k\) denotes the \(k_{th}\) entity in an entity sequence.} \(\forall k,l\leq K\): 
    \begin{equation}\label{40}
       \mathit{KB}: P(e^{k}=e_i|e^1,\ldots,e^{l}), \forall l,k \leq K,\forall e_i\in E,
    \end{equation}
    where \(K\) is the depth of the \ac{kb}.
\end{definition}

We consider two special cases: \(k\leq l\) and \(k=l+1\). 
When \(k\leq l\), it means that there is a given \(e^k\) in the \(l\)-length entity sequence. 
In this case, the KB can be used to identify substitutes. For an entity in sequence \(e^k=e_j\), 
if there is another entity \(e_i,i\ne j\) have \(P(e^k=e_i|\ldots,e^k=e_j,\ldots, e^l)=1\), then \(e_i\) can be a substitute of \(e_j\). 
When \(k=l+1\), \ac{kb} is the conditional probability of the next entity given the preceding \(k-1\) entities.

\subsection{Combine Synonyms without Semantic Ambiguity}
Assume there is a semantic probability space with the set of entities \(E\), 
each entity \(e_i\in E\) in this semantic probability space has independent semantic meaning:
\begin{equation}
    d_{\Omega}(e_i,e_j)\neq 0, \forall i \neq j.
\end{equation}
However, given a \ac{kb} of these entities, some of the entities may share the same semantic meaning: 
\begin{equation}\label{same_sem}
    (e_i|\mathit{KB})\Leftrightarrow(e_j|\mathit{KB}), i\neq j.
\end{equation}
That is, according to Definition 4 of KB, we can identify synonymous entities that satisfy: 
\begin{equation}
    P(e^k=e_i|e^k=e_j)=P(e^k=e_j|e^k=e_i)=1.
\end{equation}

\begin{figure}[htbp]
    \centering
    \includegraphics[scale=0.5]{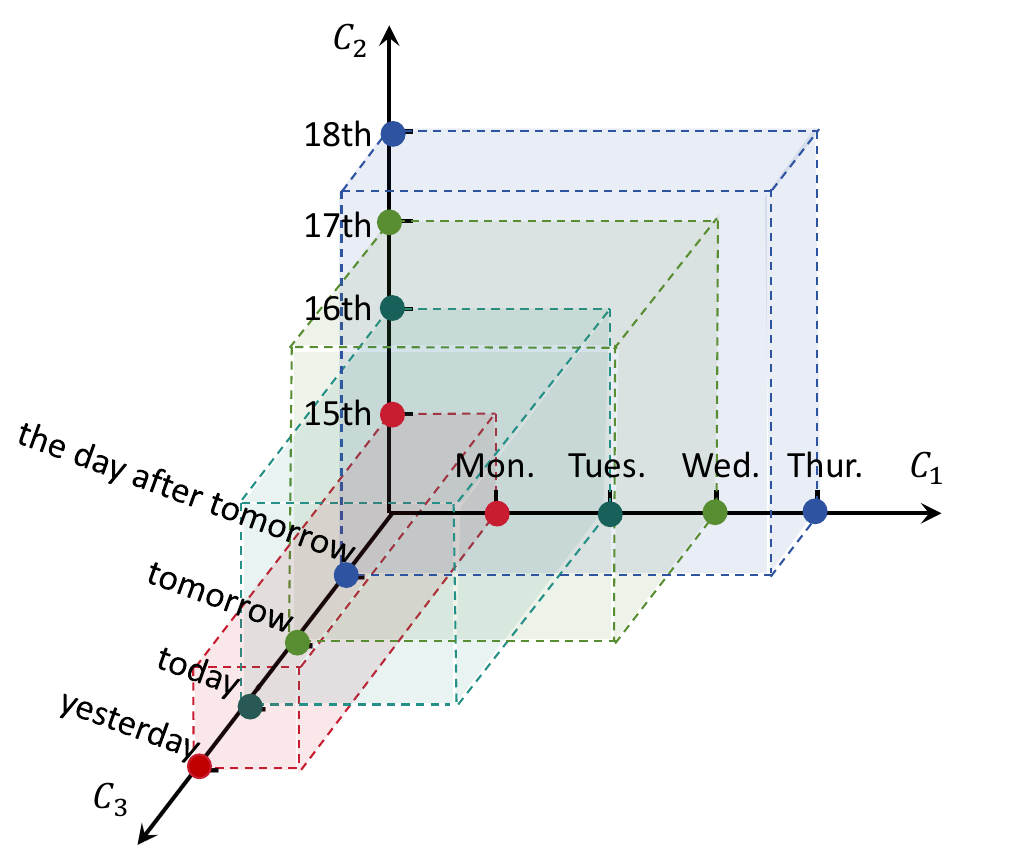}
    \caption{Combine synonyms without semantic ambiguity}
    \label{exam2}
\end{figure}

For example, in semantic probability space of date shown in Fig. \ref{exam2}, `tomorrow', `this Tuesday' and `16th' might have the same semantic meaning. 
In this case, these date entities can be combined based on calendar, as a KB. 

{Therefore, based on the relevant KB, these entities satisfying (\ref{same_sem}) can be combined to} \(e_c^{(j)}\in E_c\) without semantic ambiguity. 
From (\ref{same_sem}) we find that the entities combined to \(e_c^{(j)}\) cannot have two distinct non-zero attributes in any semantic category, otherwise there will be a semantic ambiguity. 

Let us compare the entropy of \(E\) and \(E_c\). 
In this content, the entropy before synonymous combination is calculate as:
\begin{equation}\label{H_c(E')}
        H(E)=\sum_{e_i\in E} -p(e_i)\log_2 p(e_i). 
\end{equation}
The entropy after synonymous combination is calculated as:
\begin{equation}\label{H_s(E_c)}
    \begin{aligned}
        &H(E_c)=\sum_{e_c^{(j)}\in E_c} -p(e_c^{(j)})\log_2p(e_c^{(j)})\\
        &=\sum_{e_c^{(j)}\in E_c} \left[-\sum_{e_i\in e_c^{(j)}} p(e_i)\log_2 \sum_{e_i\in e_c^{(j)}} p(e_i)\right].
    \end{aligned}
\end{equation}

\begin{proposition}
For the set of entities \(E\) in a semantic probability space, given a relevant \ac{kb}, the entities with same semantic meaning can be combined, 
and let \(E_c\) represents the set of entities after the combination. 
Then, the entropy of \(E_c\) is less than the entropy of \(E\):
    \begin{equation}
        H(E_c)\leq H(E).
    \end{equation}
\end{proposition}
\begin{proof}
    The entropy of combined entities \(E_c\) can be rewrite as: 
    \begin{equation}
    \begin{aligned}
        H(E_c)=-\sum_{e_i\in E} p(e_i)\log_2 \sum_{e_i\in e_c^{(j)}} p(e_i).
    \end{aligned}
\end{equation}
Step 1. When \(p(e_i)=0\), the corresponding sub-terms in the two equations (\ref{H_s(E_c)}) and (\ref{H_c(E')}) are the same: 
\begin{equation}
\begin{aligned}
        p(e_i)\log_2 \sum_{e_i\in e_c^{(j)}} p(e_i)=   p(e_i)\log_2 p(e_i) =0.
\end{aligned}
\end{equation}
Step 2. When \(p(e_i)>0\), we have:
\begin{equation}
    0< p(e_i) \leq \sum_{e_i\in e_c^{(j)}} p(e_i) \leq 1,
\end{equation}
thus 
\begin{equation}
     -p(e_i)\log_2 \sum_{e_i\in e_c^{(j)}} p(e_i) \leq   -p(e_i)\log_2 p(e_i).
\end{equation}
Therefore,
\begin{equation}
    H(E_c)\leq H(E).
\end{equation}
\end{proof}
P2 indicates that, based on the \ac{kb}, the source can be compressed without loss of semantic meaning.

\subsection{Scaling Semantic Attributes with Ambiguity}
Given a semantic probability space \(\left( \Omega, E,P \right)\), we can construct a new one via scaling semantic attributes. 
For simplicity, we only discuss the 1-D semantic probability space here, and the extension to higher dimensions can be derived analogously. 

For an attribute  \(a_j^k\) of \(C_j\) and its \(\epsilon_{a_j^k}\)-similar attributes, given the \ac{kb} between these attributes and an attribute \(a_j^{k'}\) of \(C_j'\): 
\begin{equation}
    P(a_j^{k'}|a_j^i)=1, \forall a_j^i\in U({a}_j^{k},\epsilon_{a_j^k}),
\end{equation}
where 
\(\epsilon_{a_j^k}\) is an acceptable similarity threshold of \(a_j^k\) based on the definition in (\ref{similar}), 
\(U({a}_j^{k},\epsilon_{a_j^k})=\{a_j^i||f_{C_j}(a_j^i)-f_{C_j}(a_j^{k})|\leq \epsilon_{a_j^k}\}\). 
Then, based on the relevant KB, the attributes of \(a_j^i\in U({a}_j^{k},\epsilon_{a_j^k})\) can be scaled to \(a_j^{k'}\): 
\begin{equation} \label{scaling}
    ({a}_j^{k'}|\Omega_{C_j'})=\left( U({a}_j^{k},\epsilon_{a_j^k})|\Omega_{C_j}\right),
\end{equation}
where \(\Omega_{C_j'}\) is a new 1-D semantic probability space obtained from scaling \(\Omega_{C_j}\). 
For example, in color category, as shown in Fig. \ref{color_scal}, there are `pink, scarlet, dark red, ruby, etc.', however, in certain contexts, the general term `red' may be sufficient. 

\begin{figure}[t]
    \centering
    \includegraphics[scale=0.5]{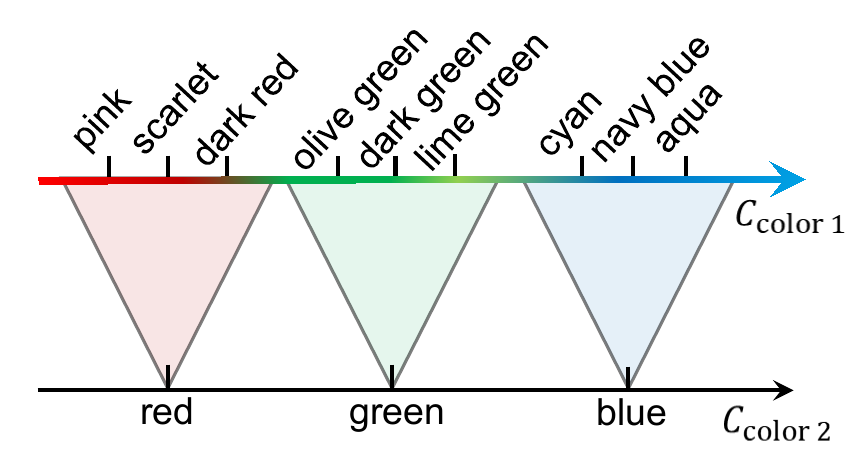}
    \caption{Scaling semantic attributes with ambiguity}
    \label{color_scal}
\end{figure}
Then, based on (\ref{similar}) and (\ref{scaling}), we give the semantic ambiguity of an attribute \(a_j^{k'}\) in category \(C_j'\) compared with \(C_j\):  
\begin{equation} \label{def_ambiguity}
    \delta_{a_j^{k’}}=\frac{\epsilon_{a_j^k}}{\textrm{len}(C_j)}, 
\end{equation}
where \(\textrm{len}(C_j)=\max_{a_j^k,a_j^l\in C_j} |f_{C_j}(a_j^k)-f_{C_j}(a_j^l)| \). 
Ambiguity here is a measurement for a semantic category, to reflect how ambiguous the semantic probability space is compared with another.

Then, we re-calculate the probability for attributes after combined in the new semantic probability space 
\begin{equation}
\begin{aligned}
    P({a}_j^{k'}|{C_j'}) &=\sum_{{a}_j^{i} \in U({a}_j^{k},\epsilon_{a_j^k})}  P({a}_j^{i}|{C_j}) \\
    &=P \left( U({a}_j^{k},\epsilon_{a_j^k})|{C_j}\right). 
\end{aligned}
\end{equation}
By scaling semantic probability space, the {entropy of attributes can} be re-calculated as: 
\begin{equation}\label{entr_ambiguity}
\begin{aligned}
    &H({C_j'})=\sum_{{a}_j^{k'}\in {C_j'}} -p({a}_j^{k'}|{C_j'})\log_2 p({a}_j^{k'}|{C_j'})\\
    &=\sum_{{a}_j^{k'}\in \Omega_{C_j'}} \left[-\sum_{{a}_j^{i}\in U_{{a}_j^{k}}} p({a}_j^{i}|{C_j} ) \log_2 \sum_{{a}_j^{i}\in U_{{a}_j^{k}}} p({a}_j^{i} |{C_j}) \right]\\
    &\leq \sum_{{a}_j^{k}\in {C_j}} -p({a}_j^{k}|{C_j})\log_2 p({a}_j^{k}|{C_j})= H({C_j}).
\end{aligned}
\end{equation}
Therefore, as expected, accepting a certain degree of semantic ambiguity can further reduce the entropy of a source information.

\section{Entropy Reduction at Message Level}
Based on the KB definition in (\ref{40}), Section III discusses how KB can reduce the entropy of source information at entity level with the condition that \(k \leq l\). In this section, we study entropy reduction at message level based on KB with the condition of \(k > l\). 






\subsection{Message Entropy}
Consider a message \(s_i\in\mathcal{S}\) consisting with \(K\) entities, i.e., \(s_i=(e^{1},e^{2},\ldots,e^{K})\).
For traditional communication with no entity relationship information (i.e., conditional probability) obtained from KB, the classical entropy of \(\mathcal{S}\) is calculated as: 
\begin{equation}
\begin{aligned}
    H_c(\mathcal{S})&=\sum_{k=1}^K H(E_k)
    \\&=\sum_{k=1}^K \sum_{e^{k}\in E_k} -p(e^{k})\log_2 p(e^{k}), 
\end{aligned}
\end{equation}
where \(E_k\) denotes the set of \(k_{th}\) entities.

For a given \(K\) depth \ac{kb}, the entropy of \(\mathcal{S}\) can be calculated with conditional probability of \(e^k\) given the preceding \(k-1\) entities. 
The express of semantic entropy of a message is 
\begin{equation}
\begin{aligned}
     H_s(\mathcal{S})=& \sum_{e^{1}\in E_1} -p(e^{1})\log_2 p(e^{2}) \\
    +&\sum_{e^{1}\in E_1}\left[ p(e^{1}) \sum_{e^{2}\in E_2} -p(e^{2}|e^{1})\log_2 p(e^{2}|e^{1}) \right] \\
    +&\cdots\\
    +&\sum_{e^{1}\in E_1} \cdot\cdot\cdot \left. \sum_{e^{{K-1}}\in E_{K-1}} \right[ p(e^{{K-1}}|e^{1}\cdots e^{{K-2}}) \\
    & \left. \sum_{e^{K}\in E_K}-  p(e^{K}|e^{1}\cdots e^{{K-1}})\log_2 p(e^{K}|e^{1}\cdots e^{{K-1}}) \right].   
\end{aligned}
\end{equation}

\begin{proposition}
   For a set of \(K\)-length messages \(\mathcal{S}\), the entropy of messages with \ac{kb}\(: P(e_{k}|e_1,\ldots,e_{k-1}), \forall k \leq K,\) is less than classical entropy, i.e., 
\begin{equation}
    H_s(\mathcal{S})\leq H_c(\mathcal{S}).
\end{equation}
\end{proposition}
\begin{proof}
Step 1. The first item of \(H_s(\mathcal{S})\) denoted by \(H_s^1(\mathcal{S})\) is the same as the first item of \(H_c(\mathcal{S})\) denoted by \(H_c^1(\mathcal{S})\): 
\begin{equation}
\begin{aligned}
    H_c^1(\mathcal{S})&= H(E_1)
    \\&=\sum_{e^{1}\in E_1} -p(e^{1})\log_2 p(e^{1}) = H_s^1(\mathcal{S}).
\end{aligned}
\end{equation}

Step 2. \(\forall i>1\) , the \(i_{\mathrm{th}}\) item of \(H_s(\mathcal{S})\) and \(H_c(\mathcal{S})\) can be calculated as :
    \begin{equation}
    \begin{aligned}
               H_s^i(\mathcal{S})&=-\sum_{e^1\in E_1} \cdot\cdot\cdot \left. \sum_{e^{i-1}\in E_{i-1}} \right[ p(e^{i-1}|e^1...e^{i-2}) \\
               &\left. \sum_{e^i\in E_i}  p(e^i|e^1...e^{i-1})\log_2 p(e^i|e^1...e^{i-1}) \right].   
    \end{aligned}
    \end{equation}
    \begin{equation}
    \begin{aligned}
               H_c^i(\mathcal{S})&=\sum_{e^i\in E_i} -p(e^i)\log_2 p(e^i)    \\
               &=-\sum_{e^i\in E_i} \left[ \cdot\cdot\cdot \sum_{e^{i-1}\in E_{i-1}}   p(e^i|e^1...e^{i-1}) \right] \log_2p(e^i)\\
               &=-\sum_{e^1\in E_1} \left[ \cdot\cdot\cdot \sum_{e^{i}\in E_{i}}   p(e^i|e^1...e^{i-1})  \log_2p(e^i)\right].   
    \end{aligned}
    \end{equation}
    Since
    \begin{equation}
        \ln{x}\leq x-1,\forall x>0.
    \end{equation}
    Then,
    \begin{equation}
    \begin{aligned}
        &H_s^i-H_c^i=\sum_{e^1\in E_1} \left[ \cdot\cdot\cdot \sum_{e^{i}\in E_{i}}   p(e^i|...)  \log_2 \frac{p(e^i)}{p(e^i|...)}\right]\\
        &\leq \sum_{e^1\in E_1} \left[ \cdot\cdot\cdot \sum_{e^{i}\in E_{i}}   p(e^i|...)\left[ \frac{p(e^i)}{p(e^i|...)} -1\right]\log_2 e \right]=0
    \end{aligned}
    \end{equation}
    Therefore,
    \begin{equation}
        H_s(S)\leq H_c(S)
    \end{equation}
\end{proof}

Proposition 3 indicates that \ac{kb} can reduce information entropy at message level by exploiting the dependencies between entities.
Taking an example to better understand, where a message , \(s=\text{`Today is Sunny'}\), then \(e^{1}=\text{`Today'}\) and \(e^{2}=\text{`Sunny'}\). 
If we have the below KB of weather,
\begin{equation}
    P(\text{`Sunny'}|\text{`Today'}) \simeq 1, 
\end{equation}
that is, given \(e^{1}=\text{`Today'}\), the probability distribution of \(E_{2}\) is
\begin{equation}
    P(e^{2}=\text{Rainy, Sunny, Cloudy,...}|e^{1}=\text{Today})\simeq (0,1,0,\cdots).
\end{equation}
If there is no KB of weather, the probability distribution of \(E_{2}\) is independent with the first entity. Assuming the below probabilities for \(E_{2}\),   
\begin{equation}
    P(\text{Rainy, Sunny, Cloudy,}\ldots)\simeq (0.1,0.2,0.2,\ldots),
\end{equation}
we can calculate the entropy of \(E_{2}\) with KB is less than without KB in this example. 

Sections III and IV theoretically validate that a KB can reduce information entropy at both entity and message levels, thereby improving the transmission efficiency in SemCom. Building on these findings, Remark 1 below reveals how \ac{ml} should be leveraged in the design of SemCom system.

\begin{remark}
ML can be leveraged in two key aspects of designing a SemCom system. 
\begin{itemize}
    \item As the depth \(K\) of the \ac{kb} increases, the number of conditional probability distributions contained in the KB grows exponentially. Estimating these conditional probabilities by applying ML facilitates better generalization capabilities in the face of unseen semantic contexts, which significantly reduces the storage and computational burdens associated with KB depth expansion.
    
    \item The code-book size for conditional probability encoding also increases exponentially with depth, storing codewords individually for each depth state can lead to codebook sizes that are difficult to manage and store. Training end-to-end coding models through ML allows optimizing the codebook to significantly reduce storage requirements.
    \end{itemize}
\end{remark}

\subsection{Semantic Channel Capacity}

Mathematically, based on the KB definition in (\ref{40}), let \(\mathbf{P}_{E_{k}}^{s^{k-1}}=\mathbf{P}(E_k|\mathcal{S}^{k-1})\) denotes the conditional probability distribution of \(E_{k}\) given \({(k-1)}\)-length sequence \(\mathcal{S}^{k-1}\), 
and let \(\mathbf{Q}_{E_{k}}\) denotes the marginal probability distribution of \(E_{k}\). 
The mutual information can be calculated as: 
\begin{equation}
\begin{aligned}
    &I_{\mathit{KB}}^E(\mathcal{S}^{k-1},E_k)=D_{\mathit{KL}} \left(\mathbf{P}(\mathcal{S}^{k-1},E_k)||\mathbf{P}(\mathcal{S}^{k-1})\mathbf{P}(E_k) \right)\\
    &=\sum_{s^{k-1}\in \mathcal{S}^{k-1}} \sum_{e_k\in E_k} p(s^{k-1},e_k) \log \frac{p(s^{k-1},e_k)}{p(s^{k-1})p(e_k)}\\
    &=\sum_{s^{k-1}\in \mathcal{S}^{k-1}} \sum_{e_k\in E_k} p(s^{k-1})p(e_k|s^{k-1}) \log \frac{p(e_k|s^{k-1})}{p(e_k)}\\
    &= \sum_{s^{k-1}\in \mathcal{S}^{k-1}} p(s^{k-1}) D_{\mathit{KL}} (\mathbf{P}_{E_{k}}^{s^{k-1}}||\mathbf{Q}_{E_{k}}).
\end{aligned}
\end{equation}
Thus, consider a \(K\)-length message the information gain obtained from \ac{kb} can be measured with divergence\footnote{\ac{kb} can be applied to any messages with the length less than \(K\) in a similar way.} 
\begin{equation}
    I_{\mathit{KB}}^M=\sum_{k=1}^K \sum_{s^{k-1}\in \mathcal{S}^{k-1}} p(s^{k-1}) D_{\mathit{KL}} (\mathbf{P}_{E_{k}}^{s^{k-1}}||\mathbf{Q}_{E_{k}}).
\end{equation}


Similarly, we can also obtain the \ac{kb}-based mutual information between a single entity and its substitutes: \(I_{\mathit{KB}}^E\). 
Then, the mutual information from \ac{kb} is:
\begin{equation}
    I_{\mathit{KB}}=I_{\mathit{KB}}^E+I_{\mathit{KB}}^M.
\end{equation}
Therefore, based on the mutual information, the gain obtained from \ac{kb} can be expressed as:
\begin{equation}
    S_\mathrm{KB}=\frac{H_c(\cdot)}{H_c(\cdot)-I_{\mathit{KB}}},
\end{equation}
where \(H_c(\cdot)\) represents classical source entropy. 
Then, the semantic channel capacity with \ac{kb} can be modeled as: 
\begin{equation}\label{sc_channel}
    C_s=S_\mathrm{KB}\frac{1}{L}B\log_2(1+\gamma),
\end{equation}
where \(B\) is the bandwidth, \(\gamma\) is \ac{snr}, and \(L=H_c(\cdot)/M\) is the attributes information density for a given \(M\)-D semantic probability space. 
The number of bits that can be transmitted is quantified by classical Shannon theory, whereas the amount of semantics that can be derived from these bits could be measured by (\ref{sc_channel}). 
Therefore, we know that the semantic channel capacity in \ac{semcom} is not only related to bandwidth and \ac{snr}, 
but also depends on the source distribution and the size of the \ac{kb}.

\begin{remark}
   The proposed mathematical framework can be utilized in the following key aspects of modeling and optimizing \ac{semcom} systems:
   \begin{itemize}
   \item From Proposition 1, we know that for a source, the categorizing entropy remains unchanged in different-dimensional spaces and under different perspective. Thus, future work could focus on determining the optimal dimension of the space to balance the trade-off between semantics and encoding efficiency, while also exploring the impact of different perspectives for locating entities in the space. 
   \item Based on (\ref{def_ambiguity}) and (\ref{entr_ambiguity}), we can optimize the scaling of the semantic space to balance the trade-off between semantic ambiguity and compression rate.
    \item Proposition 3 indicates that theoretically, the greater the depth of KB, the greater the semantic gain. However, in practice, the depth of the KB cannot be infinitely large. Therefore, we could optimize the KB to balance the trade-off between computational cost and semantics gain.
    \item Based on (\ref{sc_channel}), given the KBs of users, we can reconsider the resource optimization problem, such as bandwidth or power allocation, in wireless communication scenarios at the semantic aspect.
    
\end{itemize}
\end{remark}

\section{Numerical Results and Discussions}
In this section, numerical simulations are conducted to validate the correctness of the established SemCom mathematical framework and evaluate the SemCom system performance under different scenarios. 
Building on our SemCom mathematical model, a semantic coding is proposed, where given a perspective, using Fano coding to encode the attributes of an entity based on spatial location and probabilistic information. 
The two coding methods traditional Fano coding and Fano coding with parity check are compared as benchmarks. 
The simulation codes are built with Python 3.8 and conducted on an Intel Core i7 CPU with 16GB of RAM.

\subsection{Validity of the Proposed Framework}
Fig. \ref{sim_codingeff} demonstrates the number of entities with the coding efficiency, which is the ratio of source entropy to average length of coding an entity. 
As the number of attributes increases, the coding efficiency of all three methods increases and eventually stabilises. 
Traditional Fano is always the most efficient because it optimally allocates code words based only on the probability of occurrence of symbols without any additional redundancy. Fano with parity check is the least efficient among the three due to the additional overhead of additional parity bits for error detection. Semantic coding is slightly less efficient than traditional Fano mainly because it adds some structuring when focusing on the conceptual or attribute level for a higher level of coding mapping.

\begin{figure}[htbp]
    \centering
    \includegraphics[width=0.5\textwidth]{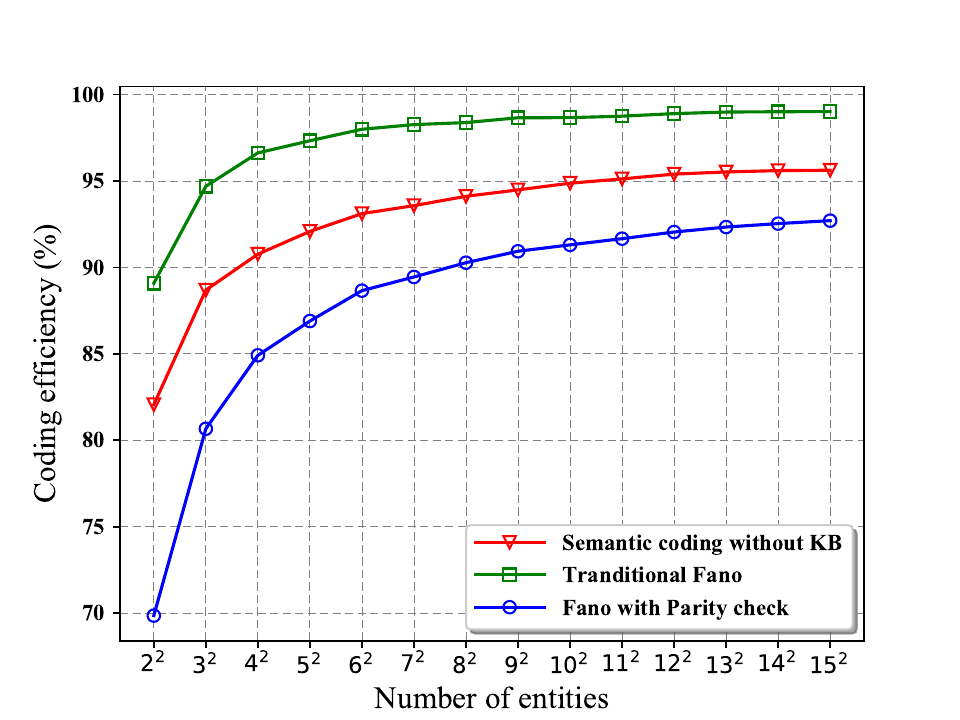}
    \caption{Coding efficiency vs. Number of attributes}
    \label{sim_codingeff}
\end{figure}

We then analyze the trend of sut error rate with varying \ac{snr} as shown in Fig. \ref{sim_ser}. From this figure, we find that the performance of all coding methods improves as the SNR increases, reflecting the positive impact of enhanced signal quality in noisy environments on coding efficiency and reliability. Specifically, traditional Fano coding exhibits higher SER at lower SNRs, and its performance rapidly improves to a superior state as the SNR increases. This is due to the fact that traditional Fano coding mainly relies on symbol occurrence probability for optimization, lacks additional error detection mechanisms, and is susceptible to high noise environments. On the other hand, Fano with parity coding exhibits lower BER at all SNR levels, due to the introduction of the parity bit which enhances the robustness and error detection capability of the system. Semantic coding has optimal basic performance under different SNR conditions. Semantic coding improves robustness under moderate to high SNR conditions by exploiting semantic correlations between attributes, making it better than traditional Fano coding in terms of BER, but not as good as Fano plus parity-check coding.

\begin{figure}[htbp]
    \centering
    \includegraphics[width=0.5\textwidth]{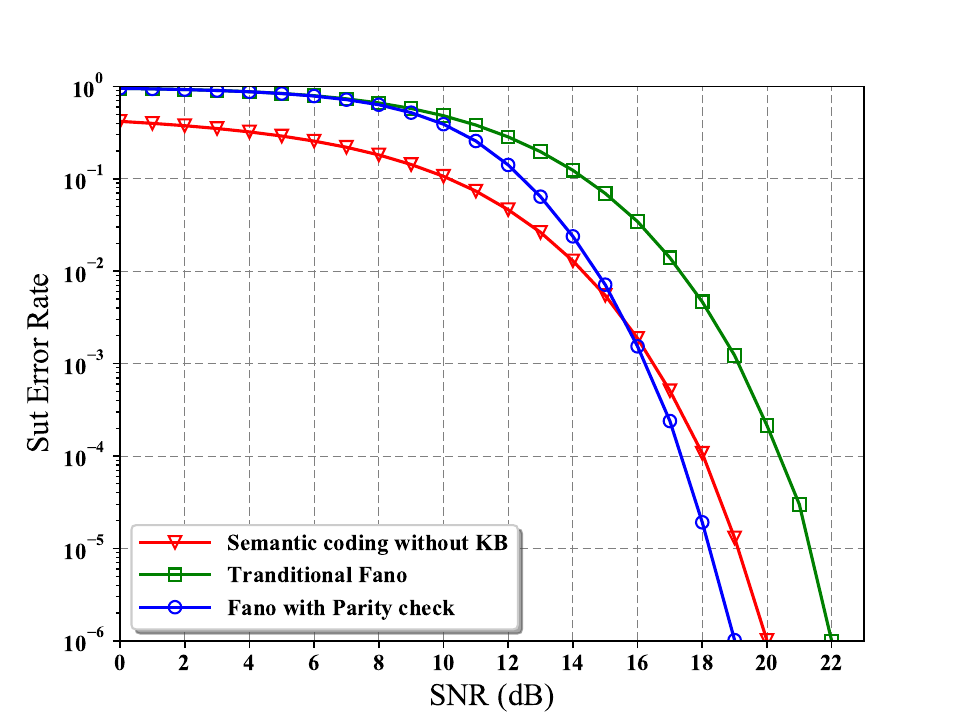}
    \caption{Sut error rate vs. Different SNR}
    \label{sim_ser}
    \end{figure}
    
We next evaluate semantic effiency under different SNR for all the three methods in Fig. \ref{sim_suteff}. 
Semantic efficiency refers to the ratio of the number of correctly transmitted semantic units to the total number of bits consumed. 
The semantic coding efficiency of all coding methods increases as the SNR increases, with semantic coding exhibiting the highest semantic coding efficiency in most SNR ranges, significantly outperforming both traditional Fano coding and Fano with parity-check coding. This suggests that semantic coding is able to make more effective use of semantic correlation and contextual information, thus achieving more efficient information transmission under high SNR conditions. Although traditional Fano coding has advantages in coding efficiency, its semantic coding efficiency is relatively low in semantic transmission due to the lack of deep mining of semantic information. And Fano plus parity check coding enhances the error detection ability by introducing parity check bits, although it improves the reliability to a certain extent, its semantic coding efficiency fails to exceed that of semantic coding due to the increase of redundant information. Therefore, the figure clearly demonstrates the superiority of semantic coding under different SNR conditions, emphasising the need for adopting semantic coding strategies in \ac{semcom} systems to enhance overall transmission efficiency and reliability.

 \begin{figure}[htbp]
    \centering
    \includegraphics[width=0.5\textwidth]{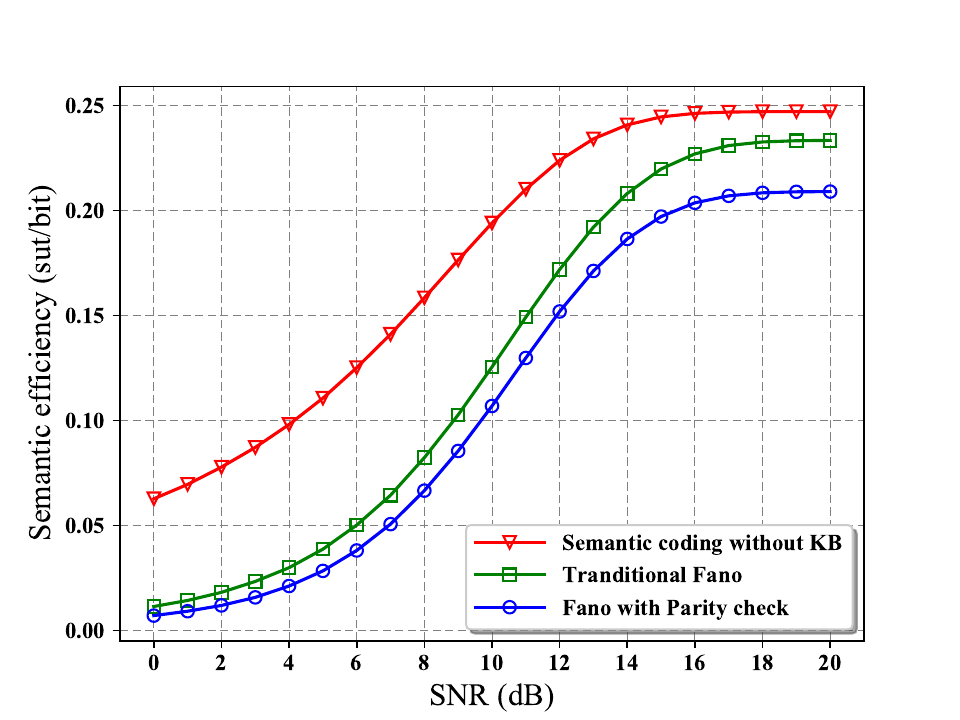}
    \caption{Semantic efficiency vs. Different SNR}
    \label{sim_suteff}
    \end{figure}   

\subsection{Validity of \ac{kb}}    
The simulation results in Fig. \ref{com_s} clearly show the relationship between the average encoding length and the number of symbols in each category. As shown, the semantic coding with KB achieves the shortest average coding length, followed by traditional Fano coding, while semantic coding without KB performs the worst in terms of coding length. The reason is that the KB helps merge semantically similar symbols, reducing redundancy and the number of effective encoding units. Without a KB, semantic coding may introduce extra structure but cannot eliminate overlaps, leading to longer codes. This confirms that incorporating semantic knowledge can significantly improve coding efficiency.    

\begin{figure}[htbp]
    \centering
    \includegraphics[width=0.5\textwidth]{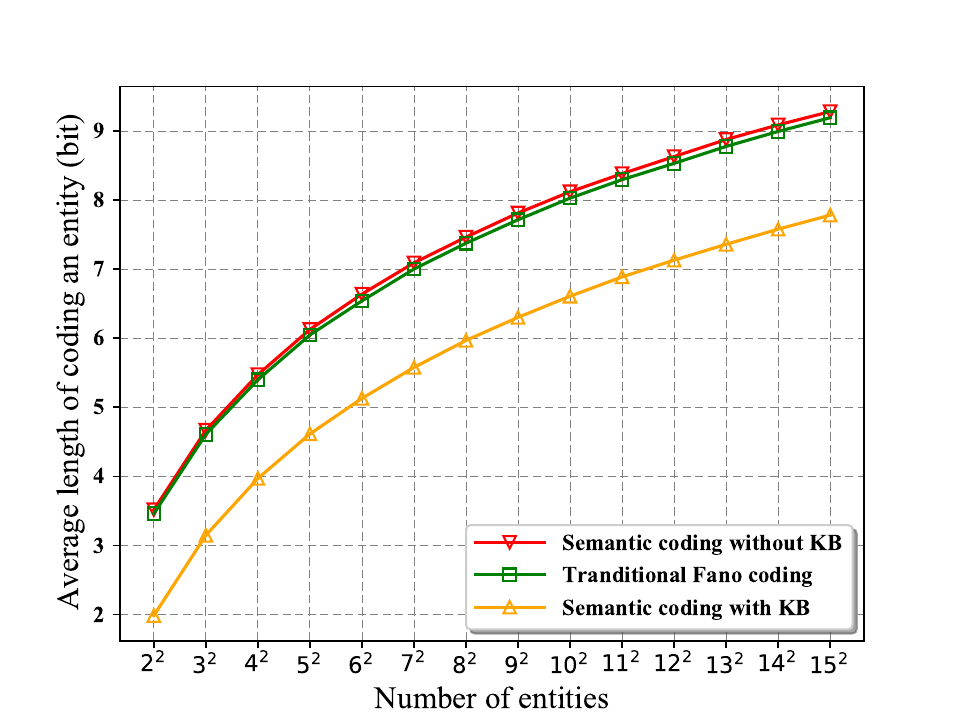}
    \caption{Average length of coding vs. Number of symbols each category}
    \label{com_s}
    \end{figure}

    \begin{figure}[htbp]
    \centering
    \includegraphics[width=0.5\textwidth]{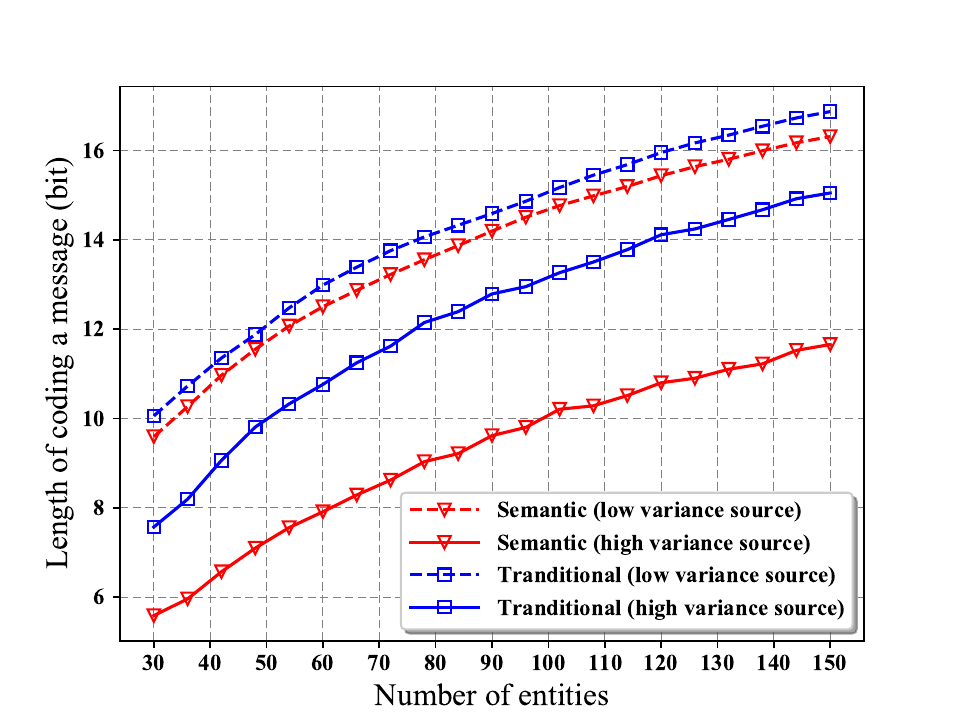}
    \caption{Length of coding vs. Number of entities each dimension}
    \label{sim_lencode}
    \end{figure}
    
Fig. \ref{sim_lencode} shows the trend of coding length with the number of entities per dimension, comparing the performance difference between the two methods with low and high variance source. As can be seen from the figure, the coding length shows an increasing trend as the number of entities per dimension increases. Under low variance source data, semantic coding significantly outperforms traditional coding, as evidenced by lower coding lengths; whereas under high variance source data, the advantage of semantic coding, although weakened, is still better than traditional coding. This shows that the semantic coding approach has a significant advantage in reducing the coding length, especially in the high variance scenario, reflecting its advantage in adapting to the data distribution characteristics.

 \begin{figure}[htbp]
    \centering
    \includegraphics[width=0.5\textwidth]{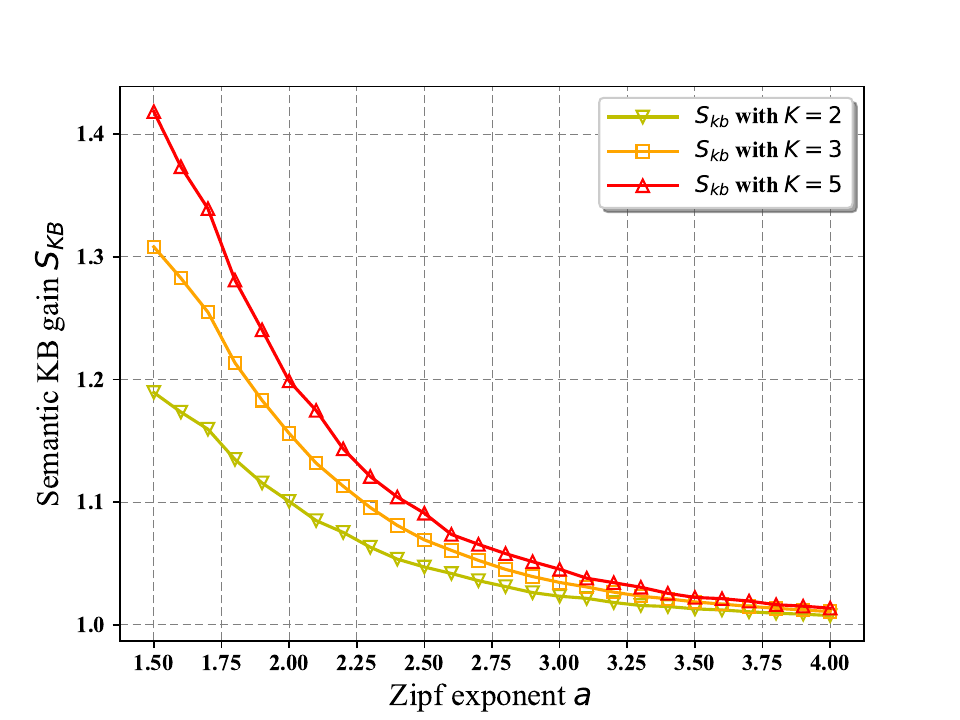}
    \caption{Semantic KB gain vs. Zipf exponent $a$}
    \label{skb}
    \end{figure}
    
Fig. \ref{skb} shows the results of \(S_\mathrm{KB}\) with exponent $a$ under discrete sources generated based on Zipf distribution. By observation we can find that \(S_\mathrm{KB}\) exhibits a significant negative correlation with the Zipf exponent $a$, and this relationship remains consistent for different values of K. As $a$ increases from 1.5 to 4.0, the \(S_\mathrm{KB}\) values of all curves show a monotonically decreasing trend and finally converge to a level close to 1.0. This is due to the basic properties of the Zipf distribution: when $a$ is small, the variance is larger, which means that the conditional probability distributions among the entities are also more diverse and interdependent. In this case, the \ac{kb} is able to provide significant inference enhancement by capturing and utilizing these conditional probability relationships. Conversely, as $a$ increases, the variance decreases and the dependencies between entities diminish. In the extreme case, almost all pairs of entities can be regarded as independent events, at which point the conditional probability information that the \ac{kb} can provide becomes negligible, and thus the KB gain approaches zero.

\section{Conclusions}
In this paper, we established a rigorous mathematical framework for \ac{semcom} leveraging category theory to define semantic entities and spaces. We introduced semantic probability space and semantic entropy to mathematically describe the semantics. We also demonstrated the role of \ac{kb} in reducing semantic entropy, enhancing communication efficiency without losing semantic meaning. Two main methods—combining synonymous entities and scaling semantic probability spaces with controlled ambiguity were analyzed, showing their effectiveness in minimizing information redundancy. Our study highlighted that semantic entropy, supported by KB, is consistently lower than traditional entropy, which significantly improves semantic channel capacity and communication efficiency. Overall, this work provides a foundational theoretical basis for future research in \ac{semcom}, particularly in optimizing semantic coding schemes and developing practical \ac{semcom} systems.

	
	\bibliographystyle{IEEEtran}
	\bibliography{ref}

\end{document}